\documentclass[a4paper,12pt]{elsarticle}
\usepackage[T1]{fontenc}
\usepackage[english]{babel}
\usepackage{ae,aecompl}
\usepackage{amsmath, amsthm}
\usepackage{amssymb}
\usepackage[utf8]{inputenc}
\usepackage[section] {placeins}
\usepackage[ruled, vlined, linesnumbered]{algorithm2e}
\newtheorem {Theorem}                 {Theorem}         [section]

\newtheorem {myalgorithm}    [Theorem]  {Algorithm}

\newtheorem {lemma}        [Theorem]  {Lemma}

\input amssym.def
\input amssym

\usepackage{pgfplots}
\usepackage{verbatim}
\usepackage{subfigure}
\usepackage{tikz}
\usetikzlibrary{shapes,arrows,backgrounds,mindmap}
\usepackage{titlesec}
\journal{arXiv}

\begin{document}
	\begin{frontmatter}
		\title{Minimum strongly biconnected spanning directed subgraph problem}
		\author{Raed Jaberi}
		
		\begin{abstract} 
			 
	     Let $G=(V,E)$ be a strongly biconnected directed graph. In this paper we consider the problem of computing an edge subset $H \subseteq E$ of minimum size such that the directed subgraph $(V,H)$ is strongly biconnected.
		\end{abstract} 
	\end{frontmatter}
	\section{Introduction}
	In $2010$, Wu and Grumbach \cite{WG2010} introduced the concept of strongly biconnected directed graphs, A directed graph $G=(V,E)$ is strongly biconnected if $G$ is strongly connected and the underlying undirected graph of $G$ has no articulation point (Note that the underlying graph of $G$ is connected if $G$ is strongly connected). Let $G=(V,E)$ be a strongly biconnected directed graph. In this paper we consider the problem of computing an edge subset $H \subseteq E$ of minimum size such that the directed subgraph $(V,H)$ is strongly biconnected.
 Observe that optimal solutions for minimum strongly connected spanning subgraph problem are not necessarily strongly biconnected, as shown in Figure \ref{figure:exampleoptimalsol}.

\begin{figure}[htp]
	\centering
	
	\subfigure[]{
	\scalebox{0.79}{
	\begin{tikzpicture}[xscale=2]
	\tikzstyle{every node}=[color=black,draw,circle,minimum size=0.9cm]
	\node (v1) at (-1.6,3.25) {$1$};
	\node (v2) at (-2.5,0) {$2$};
	\node (v3) at (-0.51, 0.5) {$3$};
	\node (v4) at (3.9,-1.5) {$4$};
	\node (v5) at (1.5,1.6) {$5$};
	\node (v6) at (5,1.7) {$6$};
	\node (v7) at (3.7,3.2) {$7$};
	\node (v8) at (4.9,0) {$8$};
	\node (v9) at (-3.43,1) {$9$};
	\node (v10) at (-3.43,3) {$10$};
	\node (v11) at (1,0.1){$11$};
	\node (v12) at (2.6,-2.1) {$12$};
	\node (v13) at (1.7,3.2) {$13$};
	
	\begin{scope}   
	\tikzstyle{every node}=[auto=right]   
	\draw [-triangle 45] (v5) to (v13);
	\draw [-triangle 45] (v13) to (v7);
	\draw [-triangle 45] (v7) to (v6);
	\draw [-triangle 45] (v6) to (v8);
	\draw [-triangle 45] (v8) to (v4);
	\draw [-triangle 45] (v4) to (v12);
	\draw [-triangle 45] (v12) to (v11);
	\draw [-triangle 45] (v11) to (v5);
	\draw [-triangle 45] (v5) to (v1);
	\draw [-triangle 45] (v1) to (v10);
	\draw [-triangle 45] (v10) to (v9);
	\draw [-triangle 45] (v9) to (v2);
	\draw [-triangle 45] (v2) to (v3);
	\draw [-triangle 45] (v3) to (v5);
	\draw [-triangle 45] (v12) to (v2);
	\draw [-triangle 45] (v7) to (v8);

	\end{scope}
	\end{tikzpicture}}
	
	}
	\subfigure[]{
	\scalebox{0.79}{
	\begin{tikzpicture}[xscale=2]
	\tikzstyle{every node}=[color=black,draw,circle,minimum size=0.9cm]
	\node (v1) at (-1.6,3.25) {$1$};
	\node (v2) at (-2.5,0) {$2$};
	\node (v3) at (-0.51, 0.5) {$3$};
	\node (v4) at (3.9,-1.5) {$4$};
	\node (v5) at (1.5,1.6){$5$};
	\node (v6) at (5,1.7) {$6$};
	\node (v7) at (3.7,3.2) {$7$};
	\node (v8) at (4.9,0) {$8$};
	\node (v9) at (-3.43,1) {$9$};
	\node (v10) at (-3.43,3) {$10$};
	\node (v11) at (1,0.1){$11$};
	\node (v12) at (2.6,-2.1) {$12$};
	\node (v13) at (1.7,3.2) {$13$};
	
	\begin{scope}   
	\tikzstyle{every node}=[auto=right]   
	\draw [-triangle 45] (v5) to (v13);
	\draw [-triangle 45] (v13) to (v7);
	\draw [-triangle 45] (v7) to (v6);
	\draw [-triangle 45] (v6) to (v8);
	\draw [-triangle 45] (v8) to (v4);
	\draw [-triangle 45] (v4) to (v12);
	\draw [-triangle 45] (v12) to (v11);
	\draw [-triangle 45] (v11) to (v5);
	\draw [-triangle 45] (v5) to (v1);
	\draw [-triangle 45] (v1) to (v10);
	\draw [-triangle 45] (v10) to (v9);
	\draw [-triangle 45] (v9) to (v2);
	\draw [-triangle 45] (v2) to (v3);
	\draw [-triangle 45] (v3) to (v5);

	\end{scope}
		\end{tikzpicture}}}
\subfigure[]{
	\scalebox{0.79}{
	\begin{tikzpicture}[xscale=2]
	\tikzstyle{every node}=[color=black,draw,circle,minimum size=0.9cm]
	\node (v1) at (-1.6,3.25) {$1$};
	\node (v2) at (-2.5,0) {$2$};
	\node (v3) at (-0.51, 0.5) {$3$};
	\node (v4) at (3.9,-1.5) {$4$};
	\node (v5) at (1.5,1.6) {$5$};
	\node (v6) at (5,1.7) {$6$};
	\node (v7) at (3.7,3.2) {$7$};
	\node (v8) at (4.9,0) {$8$};
	\node (v9) at (-3.43,1) {$9$};
	\node (v10) at (-3.43,3) {$10$};
	\node (v11) at (1,0.1){$11$};
	\node (v12) at (2.6,-2.1) {$12$};
	\node (v13) at (1.7,3.2) {$13$};
	
	\begin{scope}   
	\tikzstyle{every node}=[auto=right]   
	\draw [-triangle 45] (v5) to (v13);
	\draw [-triangle 45] (v13) to (v7);
	\draw [-triangle 45] (v7) to (v6);
	\draw [-triangle 45] (v6) to (v8);
	\draw [-triangle 45] (v8) to (v4);
	\draw [-triangle 45] (v4) to (v12);
	\draw [-triangle 45] (v12) to (v11);
	\draw [-triangle 45] (v11) to (v5);
	\draw [-triangle 45] (v5) to (v1);
	\draw [-triangle 45] (v1) to (v10);
	\draw [-triangle 45] (v10) to (v9);
	\draw [-triangle 45] (v9) to (v2);
	\draw [-triangle 45] (v2) to (v3);
	\draw [-triangle 45] (v3) to (v5);
	\draw [-triangle 45] (v12) to (v2);

	\end{scope}
	\end{tikzpicture}}}
\caption{(a) A strongly biconnected directed graph. (b) An optimal solution for the minimum strongly connected spanning subgraph problem. (c) An optimal solution for the minimum strongly biconnected spanning subgraph problem. Note that this subgraph has strong articulation points but its ungerlying graph has no articulation points}
\label{figure:exampleoptimalsol}
\end{figure}
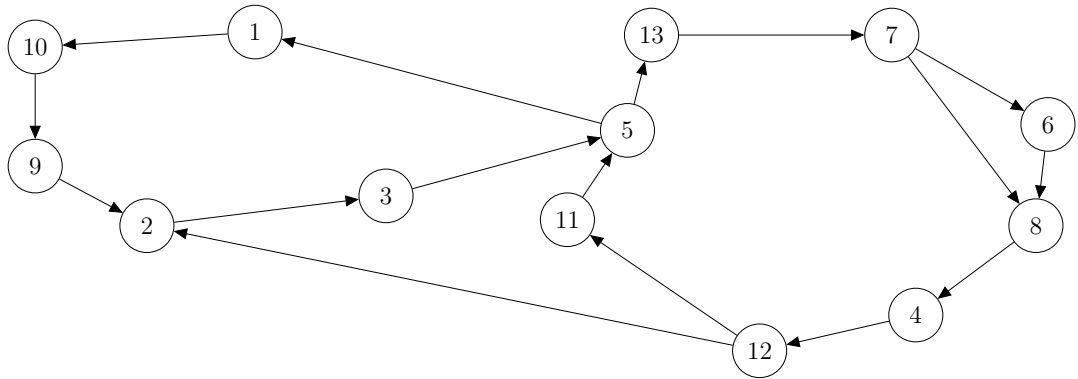
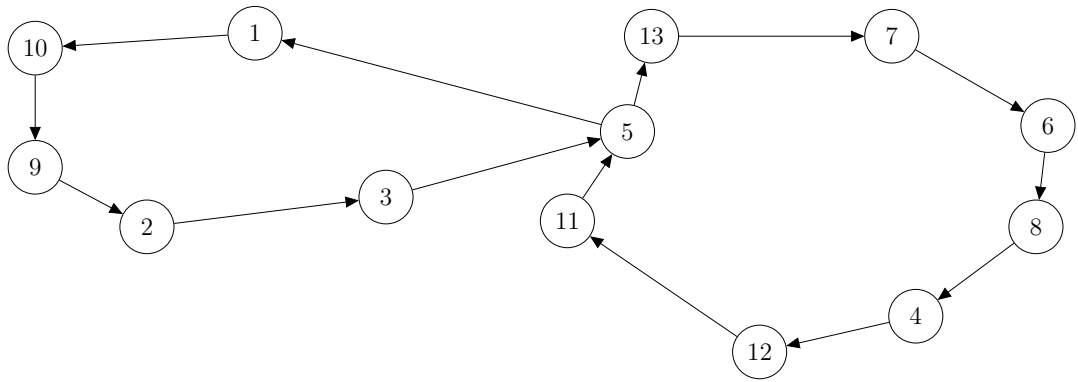
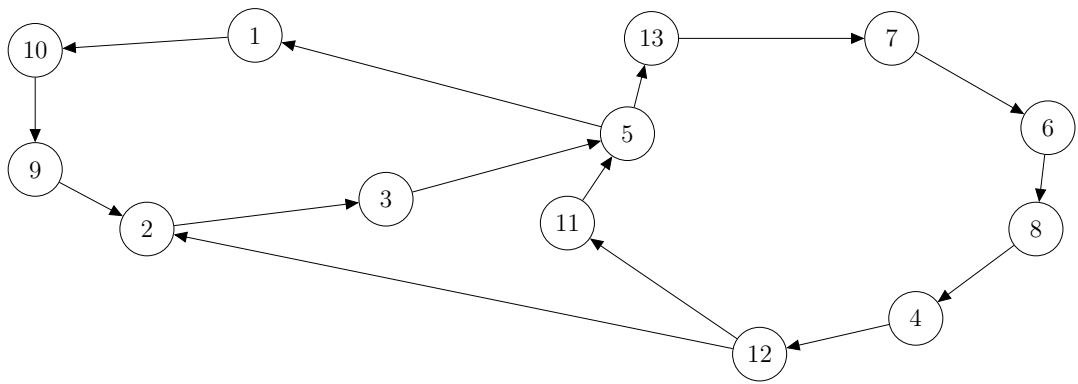

The minimum strongly connected spanning subgraph problem. is NP-complete\cite{G79}. Note that a strongly biconnected graph has a strongly biconnected spanning subgraph with $n$ edges if and only if it has a hamiltonian cycle. Therefore, the minimum strongly biconnected spanning subgraph problem is also NP-complete. Khuller et al. \cite{KRY94}, Zhao et al. \cite{ZNI03} and Vetta \cite{Vetta2001} provided approximation algorithms for The minimum strongly connected spanning subgraph problem.
Wu and Grumbach \cite{WG2010} introduced the concept of strongly biconnected directed graph and strongly biconnected components.   
Strongly biconnected directed graphs and twinless strongly connected directed graphs have been received a lot of attention in 
\cite{WG2010,Botea2018,BoteaIJCAI2018,Botea2015,Raghavan06,Jaberi2019, Jaberi21,Jaberi2021,Jaberi01897,GeorgiadisandKosinas20,Jaberi09793,Jaberi03788,Jaberi47443,Jaberi2022}.
Articulation points and blocks of an undirected graph can be calculated in $O(n+m)$ time using Tarjan's algorithm \cite{TAARJAN72,Schmidt2013}.
Tarjan \cite{TAARJAN72} gave the first linear time algorithm for calculating strongly connected components. Pearce \cite{Pearce2016} and Nuutila et al. \cite{Nuutila1994} provided improved versions of Tarjan's algorithm. Efficient linear time algorithms for finding strongly connected components were given in \cite{Sharir1981,Gabow2000,CM96, Mehlhorn2017,DietzfelbingerMehlhornSanders2014}. 
 Strong articulation points and strong bridges can be computed in linear time in directed graphs using the algorithms of Italiano et al. \cite{ILS12,Italiano2010,FGILS2016}. The algorithms of Italiano et al. \cite{ILS12,Italiano2010} are based on a strong connection between strong articulation points, strong bridges and dominators in flowgraphs. Dominators can be found efficiently in flowgraphs \cite{AHLT99,BGKRTW00,GT16,GT05,LT79}. 
In the following section we consider the minimum strongly biconnected spanning subgraph problem.

\section{Approximation algorithms for the minimum strongly biconnected spanning subgraph problem} 

In this section we study the minimum strongly biconnected spanning subgraph problem.
The following lemma shows an obvious connection between the size of an optimal solution for the minimum strongly biconnected spanning subgraph problem and the size of an optimal solution for the minimum strongly connected spanning subgraph problem.

\begin{lemma} \label{def:sizeofsbss}
Let $G=(V,E)$ be a strongly biconnected directed graph. Then the size of an optimal solution for the minimum strongly connected spanning subgraph problem is less than or equal to the size of an optimal solution for the minimum strongly biconnected spanning subgraph problem.
\end{lemma}
\begin{proof}
Let $t$ be the size of an optimal solution for the minimum strongly connected spanning subgraph problem. By definition, every  strongly biconnected spanning subgraph $G_{1}=(V,E_{1})$ of $G$ is strongly connected. Therefore, we have $\vert E_{1} \vert \geq t$. 
\end{proof}
A strongly connected spanning subgraph with size $2(n-1)$ of a strongly connected graph $G=(V,E)$ can be constructed by computing the union of outgoing branching rooted at $v \in V$ and incoming branching rooted at $v$ (\hspace{1sp}\cite{FJ81,KRY94}). But this subgraph is not necessarily strongly biconnected.
 
Algorithm \ref{algo:approximationalgorithmforsbss} can compute a feasible solution for the minimum strongly biconnected spanning subgraph problem.
\begin{figure}[h]
	\begin{myalgorithm}\label{algo:approximationalgorithmforsbss}\rm\quad\\[-5ex]
		\begin{tabbing}
			\quad\quad\=\quad\=\quad\=\quad\=\quad\=\quad\=\quad\=\quad\=\quad\=\kill
			\textbf{Input:} A strongly biconnected graph $G=(V,E)$.\\
			\textbf{Output:} a strongly biconnected subgraph $G_{v}=(V,E_{v})$ in $G$\\
			
			{\small 1}\> select a vertex $v \in V $\\
				{\small 2}\> $E_{v} \leftarrow \emptyset$\\
			{\small 3}\> build a spanning tree $T$ rooted at $v$ in $G$\\
			{\small 4}\>\textbf{for} each edge $(u,w)$ in $T$ \textbf{do} \\
			{\small 5}\>\>$E_{v} \leftarrow E_{v}\cup \left\lbrace (u,w)\right\rbrace $ \\
			{\small 6}\> build $G_{r}=(V,E_{r}),$ where $E_{r}=\left\lbrace (u,w)  \mid (w,u)\in E  \right\rbrace $\\
			{\small 7}\> construct a spanning tree $T_{v}$ rooted at $v$ in $G_{r}$\\
			{\small 8}\>\textbf{for} each edge $(u,w)$ in $T_v$ \textbf{do} \\
			{\small 9}\>\>$E_{v} \leftarrow E_{v}\cup \left\lbrace (w,u)\right\rbrace $ \\
			{\small 10}\> \textbf{while} the underlying graph of $G_{v}=(V,E_{v})$ is not biconnected \textbf{do}\\ 
			{\small 11}\>\> compute the strongly biconnected components of $G_{v}$\\
			{\small 12}\>\> find an edge $(u,w) \in E\setminus E_{v}$ such that $u,w $ do not\\
			{\small 13}\>\>\>  belong to the same strongly connected component of $G_{v}$\\
			{\small 14}\>\>  $E_{v}\leftarrow E_{v} \cup \left\lbrace (u,w) \right\rbrace  $.\\
			
		\end{tabbing}
	\end{myalgorithm}
\end{figure}
\begin{lemma} 
The output of Algorithm \ref{algo:approximationalgorithmforsbss} is strongly biconnected.
\end{lemma}	
\begin{proof}
Lines $1$--$9$ computes a strongly connected spanning subgraph $G_{v}=(V,E_{v})$ of $G$ since there is a path from $v$ to $w$ and a path from $w$ to $v$ for all $w \in V$. The while loop of lines $10$--$14$ removes all articulation points of the underlying graph of $G_{v}=(V,E_{v})$.
\end{proof}

The following lemma shows that the approximation factor of Algorithm \ref{algo:approximationalgorithmforsbss} is $3$
\begin{lemma} \label{def:optsolution2esb}
The number of edges in the output of Algorithm \ref{algo:approximationalgorithmforsbss} is at most $3(n-1)$.
\end{lemma}
\begin{proof}
Lines $1$--$9$ computes a strongly connected spanning subgraph $G_{v}=(V,E_{v})$ of size $2n-2$ since each spanning tree has only $n-1$ edges. The while loop of lines $10$--$14$ removes all articulation points of the underlying graph of $G_{v}=(V,E_{v})$ by adding at most $n-1$ edges to the subgraph $G_{v}=(V,E_{v})$ because the number of strongly biconnected components of any directed graph is at most than $n$. 
\end{proof}

\begin{Theorem}
	Algorithm \ref{algo:approximationalgorithmforsbss} runs in $O(nm)$ time.
\end{Theorem}
\begin{proof}
	A spanning tree of a strongly biconnected graph can be constructed in $O(n+m)$ time using depth first search or breadth first search. Furthermore, lines $10$--$14$  take $O(nm)$ time since the number of strongly biconnected components of any directed graph is at most $n$..
\end{proof}
 
\begin{lemma} \label{def:optsolution2esb}
Let $G=(V,E)$ be a strongly biconnected directed graph. Let $G_{1}=(V,E_{1})$ be the output of a $t$-approximation algorithm for the minimum strongly connected spanning subgraph problem for the input $G$. A strongly biconnected subgraph can be obtained from $G_{1}=(V,E_{1})$ with size at most $(1+t)h$ in polynomial time, where $h$ is the size of an optimal solution for the minimum strongly biconnected spanning subgraph problem.
\end{lemma}
\begin{proof}
A strongly biconnected subgraph can be obtained from $G_{1}=(V,E_{1})$ by running the while loop of lines $10$--$14$ of Algorithm \ref{algo:approximationalgorithmforsbss} on $G_{1}=(V,E_{1})$. By lemma \ref{def:sizeofsbss}, the size of an optimal solution for the minimum strongly connected spanning subgraph problem is less than or equal to the size of an optimal solution for the minimum strongly biconnected spanning subgraph problem. 
 This while loop takes $O(nm)$ time. The while loop of lines $10$--$14$ adds at most $n-1$ edges to $G_{1}$. Clearly, each  strongly biconnected spanning subgraph of $G$ has at least $n$ edges.
\end{proof}

The following lemma shows an obvious connection between the size of an optimal solution for the minimum strongly biconnected spanning subgraph problem and the size of an optimal solution for the minimum 2-vertex connected spanning undirected subgraph problem.

\begin{lemma} \label{def:sizeofsbssbsubgraph}
Let $G=(V,E)$ be a strongly biconnected directed graph. Then the size of an optimal solution for minimum 2-vertex connected spanning undirected subgraph problem is less than or equal to the size of an optimal solution for the minimum strongly biconnected spanning subgraph problem.
\end{lemma}
\begin{proof}
Let $t$ be the size of an optimal solution for the minimum strongly connected spanning subgraph problem. By definition, the underlying graph of every  strongly biconnected spanning subgraph $G_{1}=(V,E_{1})$ is biconnected. Therefore, we have $\vert E_{1} \vert \geq t$. 
\end{proof}

\begin{lemma} \label{def:optsolution2esb}
There is a $17/6$ approximation algorithm for the minimum strongly biconnected spanning subgraph problem.
\end{lemma}
\begin{proof}
Let $G=(V,E)$ be a strongly biconnected directed graph. A strongly biconnected subgraph can be obtained fromm
$G$ by calculating a strongly connected spanning subgraph of $G$ and a biconnected spanning subgraph of the underlying graph of $G$. Moreover, let $h$ be the size of an optimal solution for the minimum strongly biconnected spanning subgraph problem.
Let $i$ be the size of an optimal solution for the minimum strongly connected spanning subgraph problem and let $s$ be the size of an optimal solution for the minimum 2-vertex connected spanning undirected subgraph problem.
By lemma \ref{def:sizeofsbss}, we have $h\geq i$. Moreover, by lemma \ref{def:sizeofsbssbsubgraph}, we have $h\geq s$. A feasible solution of size at most $(17/6)h$ for the minimum strongly biconnected spanning subgraph problem can be obtained by running Vetta's algorithm \cite{Vetta2001} on $G$ and the algorithm of Vempala and Vetta \cite{VV00} on the underlying graph of $G$.
 
\end{proof}

\section{Open Problems}
 Results of Mader \cite{Mader71,Mader72} imply that the number of edges in each minimal $k$-vertex-connected undirected graph is less than or equal to  $kn$ \cite{CT00}. Results of Edmonds \cite{Edmonds72} and Mader \cite{Mader85} imply that the number of edges in each minimal $k$-vertex-connected directed graph is at most $2kn$ \cite{CT00}. Jaberi\cite{Jaberi47443} proved that each minimal $2$-vertex-strongly biconnected directed graph has at most $7n$ edges. The proof in is based on results of Mader \cite{Mader71,Mader72,Mader85} and Edmonds \cite{Edmonds72}.

 We leave as open problem whether the number of edges in each minimal strongly biconnected directed graph is at most $2n$ edges.
 An important question is whether there are better algorithms for the problems in \cite{Jaberi03788,Jaberi47443,DietzfelbingerJaberi2015}.

\end{document}